\documentclass[11pt]{article}

\usepackage{amsmath,amssymb,amsthm,mathtools,comment}
\usepackage{enumerate}
\usepackage{float}
\usepackage{graphicx}
\usepackage{todonotes}
\usepackage{complexity}
\usepackage{hyperref}
\usepackage{cleveref}
\usepackage{xspace}
\usepackage{authblk}
\usepackage{tikz}
\usetikzlibrary{automata,positioning, decorations.pathreplacing, arrows, decorations.markings}

\usepackage{stackengine}
\usepackage[title]{appendix}
\stackMath

\theoremstyle{definition}
\newtheorem{lemma}{Lemma}
\newtheorem{theorem}{Theorem}
\newtheorem{proposition}{Proposition}

\newtheorem{remark}{Remark}

\newcommand{\N}{\mathbb{N}}

\newcommand{\T}{\mathcal{T}}
\renewcommand{\tilde}{\widetilde}

\title{Revisiting Tree Canonization using polynomials}

\newclass{\ReachUL}{ReachUL}
\newclass{\coUL}{coUL}

\newlang{\NZCL}{NonZeroCircL}
\newlang{\NZCNC}{NonZeroCircNC}

\newclass{\Log}{L}
\newclass{\ACz}{AC^0}
\newclass{\TCz}{TC^0}
\newclass{\ACo}{AC^1}
\newclass{\ACzt}{AC^0[\oplus]}
\newclass{\FOar}{FO(\le,+,\times)}
\newclass{\FOpar}{FO[\oplus](\le,+,\times)}
\newclass{\DynACz}{DynAC^0}
\newclass{\DynTCz}{DynTC^0}
\newclass{\DynACzt}{DynAC^0[\oplus]}
\renewclass{\DynFO}{DynFO}
\newclass{\DynFOar}{DynFO(\le,+,\times)}
\newclass{\DynFOLL}{DynFOLL}
\newclass{\DynFOp}{DynFO[\oplus]}
\newclass{\DynFOpar}{DynFO[\oplus](\le,+,\times)}
\newclass{\DynFOmar}{DynFO[MAJ](\le,+,\times)}

\newclass{\LogCFL}{LogCFL}

\newlang{\PM}{PM}
\newlang{\BPM}{BPM}
\newlang{\PMD}{PMDecision}
\newlang{\BPMD}{BPMDecision}
\newlang{\PMS}{PMSearch}
\newlang{\BPMS}{BPMSearch}
\newlang{\BMWPMS}{MinWtBPMSearch}
\newlang{\MCM}{MM}
\newlang{\BMCM}{BMM}
\newlang{\BMCMD}{BMMDecision}
\newlang{\BMCMS}{BMMSearch}
\newlang{\MCMSz}{MMSize}
\newlang{\BMCMSz}{BMMSize}
\newlang{\MWMCM}{MinWtMM}
\newlang{\BMWMCM}{MinWtBMM}
\newlang{\BMWMCMS}{MinWtBMMSearch}
\newlang{\Reach}{Reach}
\newlang{\Dist}{Distance}
\newlang{\Rank}{Rank}

\begin{document}
\title{Revisiting Tree Canonization using polynomials}
\author[1,2]{V. Arvind\thanks{arvind@imsc.res.in}}
\author[2,3]{Samir Datta\thanks{sdatta@cmi.ac.in}}
\author[4]{Salman Faris\thanks{salmanfaris2048@gmail.com}}
\author[2]{Asif Khan\thanks{asifkhan@cmi.ac.in}}
\affil[1]{Institute of Mathematical Sciences (HBNI), Chennai, India}
\affil[2]{Chennai Mathematical Institute, Chennai, India} 
\affil[3]{UMI ReLaX}
\affil[4]{BITS Pilani, Hyderabad, India}
\date{}	
\maketitle
\begin{abstract}
Graph Isomorphism (GI) is a fundamental algorithmic problem. Amongst graph classes for which the computational complexity of GI has been resolved, trees are arguably the most fundamental. Tree Isomorphism is complete for deterministic logspace, a tiny subclass of polynomial time, by Lindell's result.  Over three decades ago, he devised a deterministic logspace algorithm that computes a string which is a canon for the input tree -- two trees are isomorphic precisely when their canons are identical.

Inspired by Miller-Reif's reduction of Tree Isomorphism to Polynomial Identity Testing, we present a new logspace algorithm for tree canonization fundamentally different from Lindell's algorithm. Our algorithm computes a univariate polynomial as canon for an input tree, based on the classical Eisenstein's criterion for the irreducibility of univariate polynomials. This can be implemented
in logspace by invoking the well known Buss et al. algorithm for arithmetic
formula evaluation. However, we have included in the appendix a simpler self-contained proof showing that arithmetic formula evaluation is in
logspace.

This algorithm is conceptually very simple, avoiding the delicate case analysis and complex recursion that constitute the core of Lindell's algorithm. We illustrate the adaptability of our algorithm by extending it to a couple of other classes of graphs.
\end{abstract}

\section{Introduction}
Graph Isomorphism is a classical and enigmatic problem in computer
science. On the one hand, no polynomial-time algorithm is known for
the problem (the current best being Babai's quasipolynomial time
algorithm \cite{Babai16}). On the other hand, there is not even a
$\Ptime$-hardness result known (the best hardness we know is
$\DET$-hardness by Tor\'an \cite{Toran}).

There are graph classes where the complexity gap between upper and
lower bounds has been bridged. Trees \cite{Lindell}, planar graphs
\cite{DLNTW}, interval graphs \cite{KKLV11}, and bounded treewidth
graphs \cite{ES} are well-known graph classes with matching lower and
upper bounds of logspace (the complexity class $\Log$). The logspace
algorithms for these graph classes crucially use Lindell's logspace
tree canonization algorithm \cite{Lindell} as subroutine. Here by
canonization we mean given a graph $G$ from a target class (such as
trees) create a string $\tau(G)$ such that two graphs from the target
graph class are isomorphic if and only if the corresponding strings
are identical.
%
The algorithm works by an intricate (and clever!) recursion.



An entirely different algorithm for tree isomorphism\footnote{We
  recall, in the sequential setting, that the Aho-Hopcroft-Ullman
  algorithms text \cite{AHU} presents a linear-time tree isomorphism
  algorithm.} is Miller and Reif's simple parallel algorithm
\cite{MR}. It works by reducing the problem to polynomial formula
identity testing. They obtain an arithmetic formula $\Phi_T$ for a
given rooted tree $T$ that computes a $d$-variate multivariate
polynomial $p_T$, where $d$ is the depth of $T$. Rooted trees $T$ and
$T'$ and isomorphic if and only if $p_T$ and $p_{T'}$ are identical
polynomials.
Plugging in small random values for variables in the two arithmetic
formulas yields distinct values with high probability if the trees are
not isomorphic (by the Polynomial Identity Lemma \cite{DeMilloLipton,
  Schwartz, Zippel}). The proof crucially uses the recursive
construction of irreducible multivariate polynomials for each subtrees
of the tree rooted at an arbitrary vertex. In a sense, the
Miller-Reif algorithm \cite{MR} trades complexity and determinism for
simplicity over Lindell's algorithm.

In this note we show that we need not forsake either complexity or
determinism in order to achieve simplicity for tree canonization.  To
wit, we obtain a new deterministic logspace algorithm for tree
canonization based on the Miller-Reif approach. The main idea is to
replace multiple variables by a single variable while preserving the
irreducibility of the univariate polynomial corresponding to the
subtree rooted at a vertex. At the same time the degree of the
polynomial remains upper bounded by the size of the tree. Thus the
univariate polynomial itself forms a canon for the tree. Since the
degree is polynomially bounded and we can explicitly compute the
polynomial using arithmetic formula evaluation \cite{BCGR92,HAB},
where the list of coefficients can be interpreted as a canon. This is
in contrast to the Miller-Reif multivariate polynomial $p_T$
corresponding to a tree $T$ which would require too many coefficients
to serve as a canon.


One the one hand, our approach can be viewed as a complete
derandomization of the Miller-Reif approach and on the other, it is
conceptually very simple with the nitty gritty details of polynomial
evaluation being absorbed in the algorithm to evaluate arithmetic
formulas \cite{BCGR92} and, if desired, efficient univariate
polynomial interpolation \cite{HAB}. This last is required to
interpolate the tree canon polynomial from its
evaluations. Alternatively, evaluations of the polynomial at degree
plus one many values itself serve as a canon.

Our approach yields a new proof of Lindell's result that is arguably
simpler and more conceptual. It is also adaptable to other tree-like
graph classes. We illustrate this with labelled trees, of which
block-trees are a concrete example, and also $k$-trees which are a
special case of treewidth $k$ graphs.

\section{Preliminaries}
$\Log$ is the class of all languages that are decidable by Turing machines with read only input tape and logarithmically bounded in the input size work tape. A function is said to be computable by a logspace Turing machine that takes the function input and gives its output on a separate write only tape. We say that such a function is computable in $\Log$. Functions computable in $\Log$ are closed under function composition. That is, if $f$ and $g$ are two functions ($f,g:\Sigma^*\to\Sigma^*$, $\Sigma$ is the input alphabet) that are computable in $\Log$ then $f\circ g$ is also computable in $f$ (see~\cite{LM73}). Our canonization algorithms will compose constantly many functions computable in $\Log$.

\paragraph*{Graphs and connectivity}
A graph $G$ is \textit{connected} if there is path between every pair of distinct vertices in $G$. An acyclic connected graph is a \textit{tree}. A \textit{rooted tree}, is a tree with a specified vertex called its root. Let $u$ and $v$ be vertices of a rooted tree $T$ with root $r$. Then $u$ is an \textit{ancestor} of $v$ if it lies on the path from $v$ to $r$. Equivalently, $v$ is a \textit{descendant} of $u$. In $T$, vertex $u$ is a \textit{child} of vertex $v$ if $v$ is an ancestor as well as neighbour of $u$, and $v$ is \textit{the parent} of $u$, denoted $v=\text{parent}(u)$. A vertex with no descendants is a \textit{leaf}.
   
A vertex $v$ in a connected graph $G$, is a \textit{cut vertex} if removing it makes $G$ disconnected. A graph without cut vertices is \textit{biconnected}. Maximal biconnected subgraphs of a connected graph are its \textit{biconnected components} or \textit{blocks}. The \textit{block-cut tree} of a connected graph $G$, is a tree defined from $G$ as follows. There are nodes in the block-cut tree corresponding to cut vertices of $G$ (cut nodes) and biconnected components of $G$ (block nodes). In the tree, a cut node is adjacent to a block node if and only if the corresponding cut vertex in $G$ is in the corresponding block of $G$. A \textit{block graph} is a connected graph whose biconnected components are all cliques.

A \textit{clique-sum} of two graphs, is a graph obtained after combining them by identifying the vertices in two equal-sized cliques in the two graphs. It is $k$-clique-sum if the the clique size is $k$.

A \textit{coloured graph} is a triple $(V,E,\text{col})$, where $(V,E)$ is an undirected graph and $\text{col}:V\to\N$ is a vertex colouring function.
\paragraph*{Graph isomorphism and canonization}
Graphs $G$ and $H$ are \textit{isomorphic} if there is a bijection $\phi:V(G)\to V(H)$ such that $\{\phi(u),\phi(v)\}\in E(H)\iff\{u,v\}\in E(G)$. For a graph class $\mathcal{G}$, a function $f$ defined on $\mathcal{G}$ is said to be an \textit{invariant} for $\mathcal{G}$, if for all graphs $G$ and $H$ in $\mathcal{G}$, $f(G)=f(H)$ if $G$ is isomorphic to $H$. Additionally, if $f(G)=f(H)$ implies that $G$ and $H$ are isomorphic, then $f$ is a \textit{complete invariant} for $\mathcal{G}$. If a function $f:\mathcal{G}\to\mathcal{G}$ is a complete invariant such that $f(G)$ is isomorphic to $G$ for all $G\in\mathcal{G}$, then $f(G)$ is said to be the canonical form of $G$ (under $f$), or simply the \textit{canon} for $G$. In that case, an isomorphism $\varphi$ from $G$ to its canon $f(G)$ is called a \textit{canonical labelling} (under $f$).
\paragraph*{Polynomials}
A polynomial is said to be irreducible if it cannot be factorised
further into non-trivial factors. Eisenstein's criterion, stated below,
gives a sufficient condition for irreducibility of univariate
polynomials over rationals with integral coefficients.
\begin{lemma}[Eisenstein]
	Given a univariate polynomial $Q(x)$ with integral coefficients,
	\[
		Q(x) = a_nx^n+a_{n-1}x^{n-1}+\cdots+a_1x+a_0
	\]
	$Q(x)$ is irreducible over rationals if there exists a prime number $p$ such that $p$ does not divide $a_n$, $p$ divides all of $a_{n-1},\ldots,a_1,a_0$ and $p^2$ does not divide $a_0$.
\end{lemma}   

Let $p_i$ denote the $i^{th}$ prime number for all positive integers $i$. 

\subparagraph*{Arithmetic formula} Let $\mathbb{F}$ be a field and $x_1,x_2,\ldots,x_n$ be indeterminates. Then an \textit{arithmetic circuit} over $\mathbb{F}$ is a directed acyclic graph (DAG) with unique sink vertex, and each internal vertex is labelled by a $+$ or $\times$ (also called sum and product gates respectively), while the source vertices are labelled by either the field elements or the indeterminates. Each vertex of the DAG is associated with a polynomial in $\mathbb{F}[x_1,x_2,\ldots,x_n]$, which is inductively defined as follows. For the source vertices, it is their label. For a $\times$ (or $+$) labelled vertex, it is the product (or sum) of polynomials associated with its children vertices. The \textit{size} of a circuit is the number of vertices in it. Its \textit{height} is the length of the longest path from any source to the sink vertex. An arithmetic circuit for which the underlying DAG is a tree, is an \textit{arithmetic formula}.

Let $A$ be an arithmetic formula over the field $\mathbb{F}$ and indeterminates $x_1,x_2,\ldots,x_n$. The \textit{Formula Evaluation Problem}, with input instance $A$ and a scalar tuple $(c_1,c_2,\ldots, c_n)\in \mathbb{F}^n$, is to compute the value of the formula $A$ at $(c_1,c_2,\ldots,c_n)$.

Now, we have the following after combining~\cite{BCGR92,HAB}.

\begin{theorem}{\rm\cite{BCGR92,HAB}}
\label{thm:formula}
The Formula Evaluation Problem is in $\Log$.
\end{theorem}

Indeed, there is a simpler algorithm that shows Formula Evaluation is in $\Log$ which we briefly sketch below. In Section~\ref{sec6} we present
the details in a self-contained proof.

Given as input an arithmetic formula $\Phi$ of \emph{logarithmic depth} we can apply the Ben-Or and Cleve construction \cite{BC} to transform it in logarithmic space into an equivalent constant-width arithmetic branching program, which can then be evaluated in logarithmic space. In general, given an arithmetic formula $\Phi$ (of arbitrary depth) as input, for the underlying tree of the formula we can compute a logarithmic depth \emph{working tree} (which is based on recursive tree separators) in logspace using \cite[Lemma 16]{AGGT16}. From this working tree we can compute, in logspace, an arithmetic formula of logarithmic depth that is equivalent to $\Phi$ by a gadget substitution at the tree separators.


\section{Tree canonization}
We recall Miller and Reif's reduction \cite{MR} of tree isomorphism to
polynomial identity testing (PIT). Let $T$ be a rooted tree of height
$h$. For every vertex $v$ of the tree (at some height $h'\le h$) they
construct the polynomial
$Q_v(x_1,\ldots,x_{h'}) \in \mathbb{Z}[x_1,\ldots,x_h]$ with integer
coefficients:
\[
Q_v(x_1,\ldots,x_{h'}) = 
\left\{
\begin{array}{ll}
1 & v \mbox{ is a leaf} \\
\prod_{w \textrm{ child of } v}{(x_{h'} - Q_w(x_1,\ldots,x_{h'-1}))} & \mbox{otherwise} \\
\end{array}
\right.
\]


By definition, the polynomial $Q_v$ is factorized into linear factors
(in the variable $x_{h'}$). Since $\mathbb{Q}[x_1,x_2,\ldots,x_h]$ is
a unique factorization domain, it is easy to argue by induction on the
height that two trees $T$ and $T'$ rooted at vertices $r,r'$,
respectively are isomorphic if and only if $Q_r, Q_{r'}$ are identical
polynomials.

Building on this idea, we define a \emph{univariate} polynomial
$C_T(x)$ for a rooted tree $T$ that is a complete isomorphism
invariant. We fix some notation first.  Let $r$ be the root of
$T$. For every vertex $v$ of $T$, let $T_v$ be the tree rooted at
$v$. Thus $T_r = T$. For each vertex $v$ of $T$ we will define a
polynomial $C_v(x)$. Let the children of $v$ be
$v_1,v_2,\ldots,v_{\deg(v)}$ in an arbitrary but fixed order. Let
$n_v$ denote the number of vertices in the subtree $T_v$.

We define the univariate polynomial $C_v$ inductively as follows:
\begin{equation}\label{def:B}
C_v = 
\left\{
\begin{array}{ll}
x^{n_v}+2x\prod_{i=1}^{deg(v)}{C_{v_i}(x)} +2 & \mbox{if } v \mbox{ is an internal node}\\
x+2 & \mbox{if } v \mbox{ is a leaf}\\
\end{array}
\right.
\end{equation}

An easy induction argument implies the following.
\begin{proposition}
  For each vertex $v$ of $T$, $C_v(x)$ is a univariate polynomial of
  degree $n_v$ over the rationals $\mathbb{Q}$.
\end{proposition}

The next lemma is crucial.
\begin{lemma}\label{lem:CIrred}
For every vertex $v$ of $T$, $C_v(x)$ is irreducible over $\mathbb{Q}$.
\end{lemma}
\begin{proof}
Writing $C_v(x) = \sum_{i=0}^{n_v}{C^{(i)}_v x^i}$ for integers
$C^{(i)}_v$ we observe that the highest degree coefficient
$C^{(n_v)}_v$ is odd while each $C^{(i)}_v$ is even for
$i<n_v$. Moreover, $C^{(0)}_v$ is even but not divisible by
$4$. Hence, the Lemma follows from Eisenstein's criterion
\cite{Eisenstein} applied with $p=2$.
\end{proof} 
Now we have the main lemma of this section.
\begin{lemma} \label{lem:MRDerand}
The polynomial $C_T(x)$ is a complete invariant for the rooted tree $T$.
\end{lemma}
\begin{proof}
We show this statement by structural induction. It is trivial if the tree
consists of a single leaf $v$ -- because then $C_v(x) = x+2$. Suppose the
statement holds true for each child $v_i$ of $v$. In other words, given 
$C_{v_i}$, we can reconstruct the tree $T_{v_i}$. So we just need to 
prove that given $C_v(x)$ we can uniquely infer the multiset 
$\{C_{v_1}(x), C_{v_2}(x),\ldots, C_{v_{deg(v)}}(x)\}$, because that will
allow us to reconstruct the multiset 
$\{T_{v_i}(x): v_i \mbox{ is a child of } v\}$ and hence, by hanging the 
multi-set of trees $T_{v_i}$  from the node $v$, the tree $T_v$.

But the irreducibility of $C_{v_i}(x)$ for each $v_i$, that is
guaranteed by Lemma~\ref{lem:CIrred}, ensures the unique factorization
of $\frac{C_v(x)-x^{n_v}-2}{2x}$ into $C_{v_i}$'s which are
canonical for the subtrees $T_{v_i}$s.
This completes the proof.
\end{proof}
We assume that the tree is given as input to the logspace Turing machine in its pointer representation, i.e., the list of edges of the tree.
We can now show the following.
\begin{theorem}[Lindell\cite{Lindell}]
\label{thm:Lindell}
Tree canonization is in $\Log$.
\end{theorem}
\begin{proof}
From an inspection of \cref{def:B} it is clear that the definition of 
$C_v(x)$ can be unravelled into a formula whose structure mirrors the underlying
tree. In other words, we replace each node by a sum gate of fan-in $3$ with
the first summand being a product gate with $n_v$ many $x$ gates as 
children, the second summand being a product gate with $2, x$ and the $n_v$ 
gates for the children as input and finally the constant $2$, as the third 
summand.

Our next step is to evaluate the
polynomial $C_r(x)$ (for the root $r$) at $n+1$ distinct values of
$x$, say, $1,2,\ldots,n+1$. The values
$C_r(1),C_r(2),\ldots,C_r(n+1)$ will uniquely determine $C_r(x)$,
since the degree of the polynomial $C_r(x)$ is precisely the size $n$
of the tree, and hence serve as the canon for the tree. 

This evaluation can be done in $\Log$ using the logspace
  algorithm in the alternative proof of Theorem~\ref{thm:formula}
  as described in the appendix. This algorithm first transforms the
  arithmetic formula for $C_r(x)$ into an equivalent logarithmic depth
  formula, using Lemma~\ref{lem:balFormula}. Then it applies Ben-Or and Cleve's
  construction~\cite{BC} to the balanced formula to get an equivalent
  constant width polynomial size arithmetic branching program that can be
  evaluated in $\Log$.


%
\end{proof}
To canonize unrooted trees we do the following. We consider the $n$ different rooted trees, one for each possible root, and their polynomial canon as computed by Theorem~\ref{thm:Lindell}. From these we declare the lexicographically smallest polynomial to be the canon of the unrooted tree. 
\section{Labelled tree canonization}
In this section we consider a general version of the tree canonization
problem. Nodes of labelled trees are labelled by elements from a
specified set of labels. Two labelled trees are isomorphic if there is
a label preserving isomorphism between them. Given the label set,
canonization of labelled trees is well defined.

We will give an algorithm for labelled tree canonization that assigns
a univariate polynomial as the canon for an input labelled tree. This
will be based again on Eisenstein's criterion for irreducibility of
polynomials over rationals. Without loss of generality we can assume
that the label set is a subset of natural numbers. For a rooted
labelled tree $T$, we will associate a univariate polynomial $C_v(x)$
with each node $v$, defined inductively as follows.
\[
C_v(x) = 
\begin{cases}
	x+p_{\ell(v)} & \text{ if } v \text{ is a leaf node}\\
	x^{n_v} + p_{\ell(v)}x(\prod_{u~\text{child of}~v}^{~~}C_u(x))+ p_{\ell(v)} &\text{ if }  v  \text{ is an internal node} 
\end{cases}
\]
where $p_i$ for $i\in\N$ is the $i$th prime number, $n_v$ is the
number of vertices in the subtree $T_v$, and $\ell(v)\in\N$ is the
label of the node $v\in V(T)$.

We claim that the polynomial associated with root of the tree is a
canon for the rooted labelled tree. Similar to the polynomial designed
for unlabelled trees, we use interpolation and Chinese remaindering
for the actual canon.
\begin{lemma}\label{lem:labIrr}
	For every vertex $v$ of $T$, $C_v(x)$ is irreducible over $\mathbb{Q}$. 
\end{lemma}
\begin{proof}
	We first note that the degree of $C_v(x)$ is $n_v$ for each node $v$
	of $T$. This follows from the definition of $C_v$. For a leaf node
	$v$, $C_v(x)=x+p_{\ell(v)}$ is irreducible. For an internal node
	$v$, $C_v(x)$ is irreducible by Eisenstein's criterion applied with
	prime $p_{\ell(v)}$. That is, the leading coefficient of $C_v(x)$ is
	not divisible by $p_{\ell(v)}$ (indeed, by construction it is $1 (\mod
	p_{\ell(v)})$). Furthermore, all other coefficients are divisible by
	$p_{\ell(v)}$, and the constant term is not divisible by $p^2_{\ell(v)}$.
\end{proof}

\begin{lemma}
	$C_r(x)$ is a complete invariant for the rooted tree $T$, where $r$ is the root of $T$.
\end{lemma}
\begin{proof}
	We will use structural induction for the proof. For a tree
	with a single vertex $r$, $C_r(x)$ is canonical. Now, for the
	tree $T$ rooted at $r$, let the children of $r$ be
	$v_1,v_2,\ldots,v_k$. From the definition of $C_r(x)$
	and Lemma~\ref{lem:labIrr} we have that
	$\frac{C_r(x)-x^{n_r}-p_{\ell{(r)}}}{p_{\ell{(r)}}x}$ can be
	uniquely factorized as $C_{v_1}(x)C_{v_2}(x)\ldots
	C_{v_k}(x)$, as each $C_{v_i}$ is irreducible. By the
	induction hypothesis, the polynomial $C_{v_i}(x)$ is canonical
	for the subtrees $T_{v_i}$ rooted at $v_i, 1\le i\le k$.  The
	label $\ell(r)$ of the root node $r$ itself is encoded in the
	constant term $p_{\ell(r)}$ in $C_r(x)$. Putting it together,
	it now follows that the polynomial $C_r(x)$ is a canonical
	description of the labelled tree $T=T_r$.
\end{proof}


We can now use the same logspace algorithm of Theorem~\ref{thm:Lindell} to
compute canons of labelled trees as well. The canon can be computed as
a list of pairs $(C_r(a),p_i), 0\le a\le \deg(C_r)$ for the first say
$\ell$ primes $p_i, 1\le i\le \ell$ such that $\prod_{i=1}^\ell p_i$
exceeds the coefficients of $C_r$.

Labelled tree canonization can be conveniently used to compute
canonical forms for other graph classes that have a ``tree like''
decomposition.
\subsection{Block graphs}
Recall that \textit{block graphs} are connected graphs, such that
their biconnected components are all cliques. Any connected graph is
uniquely decomposed into its biconnected components. The
\textit{block-cut tree} of $G$, with nodes as blocks and cut vertices,
encodes this decomposition. The following is a direct consequence of isomorphism invariance of connectivity.
\begin{lemma}\label{lem:blockTree}
	Any isomorphism between connected graphs $G_1$ and $G_2$ preserves
	cut vertices, blocks, block orders and paths between vertices.
\end{lemma} 
We also observe the following about isomorphism of vertex colored cliques.
\begin{proposition}\label{obs:cliqueCol}
	Two vertex colored cliques are colour-preserving isomorphic precisely
	if they have the same number of vertices with the same multiplicities
	of each colour class.
\end{proposition}
From Lemma~\ref{lem:blockTree} and Lemma~\ref{obs:cliqueCol} we give a
simple logspace computable reduction from isomorphism testing of
block graphs to labelled tree isomorphism. For a block graph $G$,
consider the labelled tree $T$ obtained from the block-cut tree of $G$
by labelling the block nodes by the corresponding clique size, and
labelling the cut nodes by $1$. Block-cut tree for any graph can be computed in $\Log$ using Reingold's connectivity algorithm~\cite{Reingold08}. However, if we are promised that the input graph is a block graph, then we can compute its block-cut tree directly, without invoking~\cite{Reingold08} in the spirit of simplicity.   
\begin{lemma}
	\label{lem:labBlock}
	Given a block graph $G=(V,E)$, its labelled block-cut tree can
	be computed in $\Log$.
\end{lemma}
\begin{proof}
	We can identify the cut vertices as follows. A vertex $v\in
	V(G)$ is a cut vertex if its neighbourhood is not a clique, instead it is a disjoint union of cliques, $C_1,C_2,\ldots,C_k$. Each clique $C_i\cup\{v\}$  constitutes a block incident on the cut vertex $v$.
	This characterizes the tree by making the cut vertex node $v$ adjacent to block node $C_i\cup\{v\}$. This is clearly in logspace.
	
	Finally, we can output the
	labelled block-cut tree by labelling the blocks with their
	sizes and labelling the cut vertices by $1$.
\end{proof}      

We want to show that block graphs $G_1$ and $G_2$ are isomorphic if and only if their labelled block-cut trees $\tilde{G_1}$ and $\tilde{G_2}$ respectively are isomorphic.   
\begin{lemma}
	\label{lem:blockCan}
	Two block graphs $G_1$ and $G_2$ are isomorphic if and only if their labelled block-cut trees $\tilde{G}_1$ and $\tilde{G}_2$ respectively are isomorphic.
\end{lemma}
\begin{proof}
	For the forward direction, let $\phi$ be an isomorphism from $G_1$ to $G_2$. This isomorphism maps a block $B_1$ of $G_1$ to a block $B_2=\phi(B_1)$ of $G_2$ and a cut vertex $c_1$ of $G_2$ to $\phi(c_2)$ of $G_2$. Further, it preserves the incidence between cut vertices and block nodes, and block orders due to Proposition~\ref{obs:cliqueCol}. Block order is the label of the corresponding block node by construction. Therefore we can define $\tilde{\phi}:\tilde{G_1}\to\tilde{G_2}$, as $\tilde{\phi}(\tilde{B_1})=\tilde{\phi({B_1})}$.
	
	For the converse direction, Let $\tilde{\phi}$ be the isomorphism map between the labelled trees $\tilde{G}_1$ and $\tilde{G}_2$. We can extend $\tilde{\phi}$ to an isomorphism $\phi$ between $G_1$ and $G_2$ as follows. Let $\tilde{B_1}$ be a block node in $\tilde{G}_1$ that is mapped to $\tilde{B}_2=\tilde{\phi}({\tilde{B}_1})$ a block node in $\tilde{G}_2$. Since their labels of $\tilde{B}_1$ and $\tilde{\phi({B_1})}$ are the same, therefore $|B_1|=|B_2|$. And the number of cut vertices incident on $B_1$ and $B_2$ is also the same, because the cut vertex nodes adjacent to $\tilde{B_1}$ are in bijection with cut vertex nodes adjacent to $\tilde{B_2}$. Cut vertex $c_1$ in $G_1$ is mapped to $c_2=\phi(c_1)$ if $\tilde{\phi}(\tilde{c}_1)=\tilde{c}_2$. Consequently, non-cut vertices in $B_1$ are in bijection with non-cut vertices in $B_2$, and therefore $\phi$ can be arbitrarily defined for the non-cut vertices between $B_1$ and $B_2$. This completes the definition of the map $\phi$. 
\end{proof}
Finally, we have the following for the block graphs. 
\begin{theorem}
	Block graph canonization is in $\Log$.
\end{theorem}
\begin{proof}
	For the block graph $G$ we can compute its labelled block-cut tree in $\Log$ due to Lemma~\ref{lem:labBlock}. Due to Lemma~\ref{lem:blockCan} two block graphs are isomorphic if and only if their labelled block-cut trees are isomorphic. Hence, the canon of the labelled block-cut tree of $G$ itself serves as a canon of $G$ which we can compute in $\Log$ using the logspace algorithm of Theorem~\ref{thm:Lindell}. 
\end{proof}

Block graphs can be alternatively characterized as $1$-clique-sum of
cliques. We have seen that canonizing such graphs is in $\Log$, but a
slightly more general class of graphs, namely $2$-clique-sum of cliques, is
as hard as general Graph Isomorphism under logspace uniform projections.
\begin{remark}
While canonization of $1$-clique-sum of cliques (block graphs) is in
$\Log$, $2$-clique-sum of cliques is as hard as Graph Isomorphism in
general via a logspace uniform projection reduction.
\end{remark}
To see this, given a graph $G$ on $n$ vertices, create a clique of
size $n$ with the same vertex labels as $G$. For every edge
$e=\{a,b\}\in E(G)$ add a vertex $v_e$ and add the edges $\{a,v_e\}$
and $\{b,v_e\}$. This is a $2$-clique-sum of $K_n$ and $|E(G)|$ many
copies of $K_3$. Call this new graph $G^*$. Then any two graphs $G_1$
and $G_2$ are isomorphic if and only if so are $G^*_1$ and $G^*_2$.

\subsection{$k$-trees}   

We recall the definition of $k$-trees by describing the inductive
process of constructing them. Every $k$-clique is a $k$-tree. Given a
$k$-tree $G$ on $n$ vertices, the $(n+1)$-vertex graph obtained by
adding a new vertex $v$ and making it adjacent to every vertex in a
$k$-clique contained in $G$, is also a $k$-tree.
Canonization of $k$-trees in $\Log$ is already known due to~\cite{ADKK12}. In this section we give an alternative proof of the same.
 
For a given $k$-tree $G$, from the set of all $k$ and $(k+1)$ size
cliques contained in $G$, a canonical tree decomposition of $G$ can be
described as follows~\cite{ADKK12}. Let $M$ be the set of
$(k+1)$-cliques in $G$, and $k$-cliques which are adjacent to more than
one $(k+1)$-cliques. Then a tree $\T$ is defined from $M$, which has
nodes corresponding to the cliques in $M$. Any two cliques $M_1$ and
$M_2$ in $M$ are adjacent in $\T$ if and only if exactly one of $M_1$
and $M_2$ is a $k$-clique which is contained in the other one. It can
be seen that $\T$, defined as above, is a tree and gives a tree
decomposition of $G$. $\T$ has a unique centre which could either be
$(k+1)$-clique or a $k$-clique.
\begin{lemma}~{\rm\cite[[Lemmas 3.3, 3.4, 4.1]{ADKK12}}
  The tree $\T$ defined
  from $k$ and $(k+1)$-cliques of the $k$-tree $G$, is a tree
  decomposition of $G$, and $\T$ has a unique centre. $\T$ can be computed in $\Log$.  
\end{lemma}
Similar to~\cite{ADKK12}, for a vertex $v$ not in the centre of $\T$
(which is a clique), let $M_v$ denote the $(k+1)$-clique node
containing $v$, closest to the centre of $\T$. We say that vertex $v$
is \textit{introduced} in the node $M_v$. Given a labelling of the
vertices in the centre of the tree $\T$, we can obtain a labelled tree
for $G$ which is isomorphism invariant, as long as the labels of the
vertices in the centre are preserved. Then, akin to ~\cite{ADKK12}, we
try out all labellings of the (at most $k+1$) vertices in the centre to
choose one as a canon for the graph $G$.

Before we describe our labelling algorithm, it is useful to have a
$(k+1)$-colour proper colouring of $G$. We root the tree $\T$ at its
centre $R$. Let $|R|=k'\in\{k,k+1\}$. As $R$ is a $k'$-clique, all its
vertices must get $k'$ distinct colours, from $\{1,2,\ldots,k+1\}$, in
any proper colouring of $G$. There are at most $(k+1)!$ many such
proper colourings of $R$. We describe a simple procedure that takes
such a colouring of $R$ and extends it to a unique $k+1$-colour proper
colouring of $G$.  The set of colours is $\{1,2,\ldots,k+1\}$.

\begin{enumerate}
\item We start with a proper colouring of the vertices in the root node
  $R$ of the tree $\T$. We will extend the colouring to the remaining
  nodes of the graph $G$ by following root downwards the decomposition
  tree $\T$ starting with $R$. 
\item In a general step suppose $M$ is a $(k+1)$-clique that is properly
  coloured. Then its children are $k$-cliques contained in $M$. In
  particular, if $S$ is a child of $M$, $|S|=k$ and the vertices
  contained in $S$ get the same colour as they have in $M$.
\item If $M$ is a $k$-clique then its children are $(k+1)$-cliques $S$
  such that $M$ is contained in $S$. The unique vertex in
  $S\setminus M$ gets the unique colour not used in $M$, and the $k$
  vertices in $S\cap M$ inherit the same colour as in $M$.
\end{enumerate}
The above implicitly defines a graph $G_\text{col}$ on $V(G)$, such that two vertices $u$ and $v$ are adjacent if the symmetric difference of two consecutive $k+1$-cliques is $\{u,v\}$. This is clearly acyclic, and all the nodes in the same connected component get the same colour as the colour of the highest vertex in that connected component. Thus the $(k+1)$-colouring can be inferred in $\Log$.

\begin{lemma}
  For each proper colouring of the vertices in the root $R$ of $\T$,
  the vertex colouring defined by the above procedure is a proper
  colouring of $G$ using colours $\{1,2,\ldots,k+1\}$. Moreover, this
  colouring is isomorphism invariant, as long as the indices of the
  vertices in the root clique are preserved. The colouring can be computed in $\Log$.
\end{lemma}

\begin{proof}
  That we obtain a $(k+1)$-colouring of $G$ is obvious. We claim that for
  every clique, the vertices in it are properly coloured. We prove this
  by structural induction on the tree, with the base case being the
  root node. This is clearly true for the root node. For any internal
  $k$-clique node, since it is a subset of the vertices in the parent
  $(k+1)$-clique, its vertices continue to have the same colour as they
  have in the parent $(k+1)$-clique. For an internal $(k+1)$-clique
  node $M_v$, the introduced vertex $v$ by definition gets a distinct
  colour that is not taken up by any of the vertices in the parent of
  $M_v$, and hence $M_v$ is properly coloured. Recall that, $\T$ is a
  tree decomposition as well. So, the endpoints of each edge lie in in
  some common node of $\T$ which is properly coloured. Hence $G$ is
  properly coloured.
	
  Since the colouring is uniquely determined by the colouring of $R$,
  the indexing of the vertices of the root clique node, and the tree
  decomposition $\T$ which is isomorphism invariant, it follows that
  the entire vertex colouring is preserved by isomorphisms which
  preserve the colours of vertices in the root node of the
  decomposition tree.
\end{proof}

With this colouring, generated as above by a colouring $col$ of $R$, we
label nodes of $\T$ as follows: the $(k+1)$-clique nodes are labelled
$k+2$, the $k$-cliques are labelled by the colour missing in it, and
the root node is specially marked $r$. Let $\T_{col}$ denote this
labelled tree rooted at $r$.  We give it as input to the labelled tree
canonization routine to obtain a candidate canon. We select the
`smallest' canon obtained from the different ($(k+1)!$ many) 
proper colourings of $R$. This completes the proof of:

\begin{theorem}[\cite{ADKK12}]
	$k$-trees can be canonized in $\Log$.
\end{theorem}

\section{Conclusion}
We present an alternative proof to Lindell's tree canonization in $\Log$
that is conceptually simpler and more structured though not elementary,
since we additionally need arithmetic formula evaluation in $\Log$. It is also
easily adaptable to canonization of tree-like classes as demonstrated for
block graphs and $k$-trees. We have deliberately refrained from extending 
the approach to well known graph classes with  canonical tree-decompositions 
like planar graphs and interval graphs in order to retain the simplicity of the
core method. Canonizing such classes through polynomials remains a future goal.

\bibliography{main}
\begin{appendices}

\section{Proof of Theorem~\ref{thm:formula}} \label{sec6}
\paragraph{Tree balancing}(\cite[Section 3.4]{AGGT16})
We first recall the algorithm from~\cite{AGGT16} that computes a
recursively balanced tree separator for a given tree.


Consider a rooted tree $S$. For any node $v$ of $S$, let $S_v$ denote the subtree of $S$ rooted at $v$. For nodes $v$ and $v'$ such that $v$ is ancestor of $v'$, let $S_{v,v'}$ denote the tree $S_v\setminus S_{v'}$. For $S_v$, let $c$ be a node such that removing it from $S_v$ breaks $S_v$ into connected components, each bounded by $|S_v|/2$ in size. We define $c$ to be the centre of $S_v$ (pick $c$ to be the lexicographically smallest such node if there is a choice). Let the children of $c$ be $c_1,c_2,\ldots c_k$. Then after removal of $c$ from $S_v$, the $k+1$ connected components\footnote{In case $v=c$ there
  are only $k$ subtrees of $c$.} are $S_{c_1},S_{c_2},\ldots,S_{c_k}$ and $S_{v,c}$ each of size bounded by $|S_v|/2$. The \emph{working tree} corresponding to $S_v$, denoted $wt(S_v)$, is rooted at $c$ and is defined inductively to be the tree obtained by making the roots of the working trees $wt(S_{c_1}),wt(S_{c_2}),\ldots,wt(S_{c_k})$ and $wt(S_{v,c})$ as the children of $c$. With the node $c$ in $wt(S_v)$ we associate the tree $S_v$.  

We now define the centre for a subtree of the form $S_{v,v'}$. On the path from $v$ to $v'$, pick the first node $u$ such that $|S_{u,v'}|\le 1/2|S_{v,v'}|$. Such a vertex must exist. Because initially, for $u=v$, we have $|S_{v,v'}|>1/2|S_{v,v'}|$, and finally, for $u=v'$, we have $|S_{v',v'}|=0$. We define the parent $c$ of $u$ to be the centre of $S_{v,v'}$. Let the remaining children of $c$ be $c_1,c_2,\ldots,c_k$. The connected components of $S_{v,v'}$ after removing $c$ are $S_{c_1},S_{c_2},\ldots,S_{c_k},S_{u,v'}$ and $S_{v,c}$. As $u$ is the first vertex on the $v$ to $v'$ such that $|S_{u,v'}|\le 1/2|S_{v,v'}|$, it follows that  $|S_{c,v'}|\ge 1/2|S_{v,v'}|$. Moreover, $|S_{v,c}|+|S_{c,v'}|=|S_{v,v'}|$. Hence, $|S_{v,c}|\leq 1/2|S_{v,v'}|$. Thus, two of the connected components of $S_{v,v'}$, after removal of $c$, namely $S_{v,c}$ and $S_{u,v'}$ are balanced. That is, both $|S_{v,c}|$ and $|S_{c,v'}|$ are bounded by $|S_{v,v'}|/2$. Next, we recursively compute the working trees for the remaining connected components $S_{c_1},S_{c_2},\ldots,S_{c_k}$. Overall, the \emph{working tree} for $S_{v,v'}$ is rooted at $c$ and is inductively defined to be the tree obtained by making the roots of the working trees $wt(S_{c_1}),wt(S_{c_2}),\ldots,wt(S_{c_k}),wt(S_{v,c})$ and $wt(S_{u,v'})$ the children of $c$. With the node $c$ in $wt(S_{v,v'})$ we associate the tree $S_{v,v'}$.

By construction, we have that the depth of $wt(S)$ is $O(\log n)$, because the size of the associated tree of a node of the working tree is at most half the size of the tree associated with its grandparent node.

\begin{remark}
  If $S$ is a binary tree then its working tree $wt(S)$, by the above construction, will be ternary (each node has at at most three children).
\end{remark}  

\begin{lemma}(\cite[Lemma 16]{AGGT16})\label{wt-lem}
For any tree $S$, its working tree, $wt(S)$ as defined above, can be computed in logspace.
\end{lemma}
      
\begin{proof}
We can compute the working tree for a tree $S$ in logspace as follows. First, by a standard logspace traversal of the tree, we compute the size of every subtree of $S$. To find the parent of a given node $d$ in the working tree $wt(S)$, we run the following logspace procedure. Find the centre of $S$, let it be $c$. Removing the centre we get many components, find the one that contains the vertex $d$, call it $S_1$. Apply the same procedure recursively on $S_1$, getting deeper in smaller and smaller components that contain $d$, until $d$ itself becomes the centre of some component. The centre of the previous component containing $d$, in the recursion, is the parent of $d$ in the working tree. In this recursive procedure, we need to keep track of the current component, which requires storing at most two nodes, e.g., $v$ and $v'$ for $S_{v,v'}$. Moreover, we need to store the centre of the last component in the recursion. 
To find the centre of a tree of the form $S_v$, we can try all the vertices of the tree as potential centre and check if they satisfy the size requirements. To find the centre of a tree of the form $S_{v,v'}$, we need only try the vertices on the $v$ to $v'$ path in the tree $S_{v,v'}$ as potential centres.        
\end{proof}

\paragraph{Depth-reduction of arithmetic formulas in logspace} We can assume that the underlying tree of the input arithmetic formula is a binary tree $A_r$ (rooted at the output gate $r$). That is, every gate ($+$ or $\times$) in it has two inputs. The depth-reduction algorithm is a recursive procedure based on
the working tree construction described above (Lemma~\ref{wt-lem}). 

For a gate $g$ in $A_r$ let $A_g$ denote the subformula rooted at $g$. If we
replace the subformula $A_g$ with a new variable $x_g$, making $g$ an
input gate, the resulting formula $A_{r,g}$ is \emph{linear} in the variable $x_g$. Hence, we can write
\[
  A_{r,g}=Bx_g+C,
\]
and note that $A_r=B\cdot A_g+C$. Formula $A_{r,g}$ is obtained by \emph{scarring} $A_r$ at gate $g$. The \emph{scar} at $g$ in the formula $A_r$ sets up the recursive depth-reduction procedure that we now explain.

For an arithmetic formula $\Phi$, let $\hat{\Phi}$ denote the equivalent
depth-reduced formula computed by the procedure. That is, $\Phi\equiv \hat{\Phi}$.

Given $\hat{A}_g$, we recursively compute $\hat{A}_r$ as follows. Let $\hat{A}_{r,g}[x_g=0]$ be the depth-reduced arithmetic formula equivalent to $A_{r,g}$ with $x_g$ set to $0$. The formula $\hat{A}_{r,g}[x_g=1]$ is similarly defined. Notice that $\hat{A}_{r,g}[x_g=0]\equiv C$ and $\hat{A}_{r,g}[x_g=1]\equiv B+C$,
and hence we have
\[
  A_r \equiv ((\hat{A}_{r,g}[x_g=1]-\hat{A}_{r,g}[x_g=0])\times \hat{A}_g)+\hat{A}_{r,g}[x=0].
\]
Now, to obtain $\hat{A}_r$, we will invoke the working tree construction
corresponding to the rooted binary tree of $A_r$. The centres identified
by the working tree construction will be the gates to scar in the corresponding
subformulas.


For the subformula $A_v$ rooted at a gate $v$, let $c$ be a balanced separator of $A_v$ picked as centre by the working tree algorithm of Lemma~\ref{wt-lem}. We denote by $\circ_c$ the gate corresponding to vertex $c$. Let the children of $c$ in $A_v$ be $c_1$ and $c_2$. Then the constant-size arithmetic gadget defined at the centre $c$ that computes $A_v$ from the recursively obtained depth-reduced formulas $\hat{A}_{v,c},\hat{A}_{c_1},\hat{A}_{c_2}$ is described
below. Notice that this gadget is a circuit (and not a formula) but of
constant size.

\begin{align}
	A_v \equiv &((\hat{A}_{v,c}[x_c=1]-\hat{A}_{v,c}[x_c=0])\times(\hat{A}_{c_1}\circ_c \hat{A}_{c_2}))+ \hat{A}_{v,c}[x_c=0]  
\end{align}

If the tree associated with vertex $c$ is $A_{v,v'}$, then we do the following. By the working tree construction, vertex $c$ lies on the path from $v$ to $v'$ in $A_{v,v'}$. Let $c_1$ and $c_2$ be the children of $c$ in $A_{v,v'}$ such that $v'$ is a descendant of $c_2$. Then the following arithmetic gadget at centre $c$ defines two outputs: $A_{v,v'}[x_{v'}=\theta]$ for $\theta\in\{0,1\}$.
\begin{align}
	A_{v,v'}[x_{v'}=\theta] \equiv & ((\hat{A}_{v,c}[x_c=1]-\hat{A}_{v,c}[x_c=0])\times(\hat{A}_{c_1}\circ_c \hat{A}_{c_2,v'}[x_{v'}=\theta]))\nonumber\\
	& +\hat{A}_{v,c}[x_c=0]
\end{align}

In both cases, the gadget defined at centre $c$ is a circuit of constant size and constant depth. Suppose the inputs to the gadget are already converted into formulas, replacing the gadget at $c$ by an equivalent formula will incur a constant factor blow-up in size at $c$.  As the working tree has $O(\log s)$ depth, where $s$ is the number of gates in $A_r$, replacing the circuit
gadget by an equivalent formula at each centre $c$ will result in an $s^{O(1)}$
size blow-up in the obtained logarithmic depth formula $\hat{A}_r$.

Clearly, because of the local gadget substitution, the depth-reduced formula can be computed in logspace along with the working tree construction. Thus we have the following.




\begin{lemma}\label{lem:balFormula}
	Given an a arithmetic formula $\Phi$ of size $s$ of arbitrary depth, an equivalent arithmetic formula $\hat{\Phi}$ of depth $O(\log s)$ and size $s^{O(1)}$ can be computed in $\Log$.  
\end{lemma}

Having computed the balanced arithmetic formula, we can apply the Ben-Or and Cleve construction \cite{BC} to transform it in logarithmic space into an equivalent constant-width arithmetic branching program, which can then be evaluated in a straightforward manner in logarithmic space. This completes the proof of the Theorem~\ref{thm:formula}.

\end{appendices}
\end{document}